\documentclass[10pt, doublecolumn]{IEEEtran}
\usepackage{amsmath}
\usepackage{bbm}
\usepackage{graphicx}
\usepackage{epsfig}
\usepackage{array}
\usepackage{psfrag}
\usepackage{amssymb}
\usepackage{mathdots}
\usepackage{color}
\usepackage{pstricks,pst-node,pst-text,pst-3d,pst-plot}
\usepackage{psfrag}
\usepackage{enumerate}
\usepackage{url,cite}
\usepackage{amsfonts}%
\usepackage{amsthm}
\usepackage{dsfont}%
\usepackage{verbatim}%
\usepackage{setspace}
\usepackage{float}
\usepackage{url}
\usepackage{fancyhdr}
\usepackage{bm}
\usepackage{lipsum} %\footnote
\newcommand\blfootnote[1]{%
  \begingroup
  \renewcommand\thefootnote{}\footnote{#1}%
  \addtocounter{footnote}{-1}%
\endgroup
}
\usepackage{algorithm}
\usepackage{algorithmic}

\usepackage{cleveref}
\newcommand{\Bmax}{{B_{\rm max}}}

\newcommand{\Paux}{P_{\rm X}}
\newcommand{\Psiaux}{\Psi_{\rm X}}

\newcommand{\Oneik}{\mathds{1}_i(k)}
\newcommand{\Ai}{A_i}

\newcommand{\phimax}{\phi_{\rm max}}
\newcommand{\Rik}{R_i(k)}

\newcommand{\Qik}{Q_i(k)}

\newcommand{\Yik}{Y_i(k)}
\newcommand{\bfYk}{{\mathbf Y}(k)}
\newcommand{\bfQk}{{\mathbf Q}(k)}
\newcommand{\LW}[1]{W_0\lb #1e^{-1}\rb}

\newcommand{\gammaik}{\gamma_i(k)}

\newcommand{\NR}{N_{\rm R}}
\newcommand{\NNR}{N_{\rm NR}}

\newcommand{\sN}{\script{N}}
\newcommand{\Nk}{N_k}
\newcommand{\sNk}{\script{N}_k}
\newcommand{\sNR}{\script{N}_{\rm R}}

\newcommand{\sNNR}{\script{N}_{\rm NR}}

\newcommand{\sS}{\script{S}}

\newcommand{\SRT}{\script{S}_{\rm RT}}

\newcommand{\SRk}{\script{S}_{\rm R}\lb k\rb}
\newcommand{\SRstk}{\script{S}_{\rm R}^*\lb k\rb}

\newcommand{\SRkst}{\script{S}_{\rm R}^*\lb k\rb}
\newcommand{\SNRk}{\script{S}_{\rm NR}\lb k\rb}
\newcommand{\bfPRk}{\bfP\lb k\rb}

\newcommand{\PRik}{P_i\lb k\rb}
\newcommand{\PRikst}{P_i^*\lb k\rb}

\newcommand{\PNRik}{P_i\lb k\rb}
\newcommand{\bmuRi}{\overline{R}_i}

\newcommand{\bfmuRk}{{\bm \mu}\lb k\rb}

\newcommand{\muRik}{\mu_i\lb k\rb}

\newcommand{\muRistk}{\mu_{\ist}\lb k\rb}

\newcommand{\parPRdef}[1]{\triangleq \left[#1_i(k)\right]_{i\in\sN}}
\newcommand{\parPNRdef}[1]{\triangleq \left[#1_i(k)\right]_{i\in\sNNR}}

\newcommand{\PsiRik}{\Psi_{\rm R}(i,k)}

\newcommand{\PsiNRik}{\Psi_{\rm NR}(i,k)}
\newcommand{\PsiNRikst}{\Psi_{\rm NR}^*(i,k)}
\newcommand{\PsiNRistkst}{\Psi_{\rm NR}^*(\ist,k)}

\newcommand{\ist}{i_{\rm NR}^*}

\newcommand{\bfrk}{{\bm r}\lb k\rb}

\newtheorem{thm}{Theorem}
\newtheorem{lma}{Lemma}
\DeclareMathOperator{\E}{\mathbb{E}}

\newcommand{\lb}{\left (}
\newcommand{\rb}{\right )}

\newcommand{\script}[1]{{\mathcal {#1}}}
\newcommand{\Pavg}{P_{\rm avg}}
\newcommand{\Pmax}{P_{\rm max}}

\newcommand{\EE}[1]{\E \left[ #1 \right]}
\newcommand{\EEU}[1]{\E_{\bfU(k)} \left[ #1 \right]}

\newcommand{\bfP}{{\bf P}}

\newcommand{\bgamma}{\overline{\gamma}}

\newcommand{\bfQ}{{\bf Q}}

\newcommand{\bfY}{{\bf Y}}
\newcommand{\bfU}{{\bf U}}

\newcommand{\parRdef}[1]{\triangleq [#1_1(k),\cdots,#1_{\NR}(k)]^T}
\newcommand{\parNRdef}[1]{\triangleq [#1_1(k),\cdots,#1_{\NNR}(k)]^T}

\newcommand{\Rmax}{R_{\rm max}}
\newcommand{\gammamax}{\gamma_{\rm max}}

\newcommand{\Ts}{T}

\begin{document}
\title{Optimal Power Control and Scheduling under Hard Deadline Constraints for Continuous Fading Channels}

\author{Ahmed Ewaisha, Cihan Tepedelenlio\u{g}lu\\
\small{School of Electrical, Computer, and Energy Engineering, Arizona State University, USA.}\\
\small{Email:\{ewaisha, cihan\}@asu.edu}\\
%\small{December 2015}\\
%\thanks{}
%\today
}
\maketitle
\blfootnote{The work in this paper has been supported by NSF Grant ECCS-1307982.}
%\blfootnote{This work is in submission to IEEE Transactions on Vehicular Technology \cite{Ewaisha_TVT2017}.}%The authors are with the Electrical Engineering department at Arizona State University, Tempe, Arizona.}
\begin{abstract}
We consider a  joint scheduling-and-power-allocation problem of a downlink cellular system. The system consists of two groups of users: real-time (RT) and non-real-time (NRT) users. Given an average power constraint on the base station, the problem is to find an algorithm that satisfies the RT hard deadline constraint and NRT queue stability constraint. We propose a sum-rate-maximizing algorithm that satisfies these constraints. We also show, through simulations, that the proposed algorithm has an average complexity that is close-to-linear in the number of RT users. The power allocation policy in the proposed algorithm has a closed-form expression for the two groups of users. However, interestingly, the power policy of the RT users differ in structure from that of the NRT users. We also show the superiority of the proposed algorithms over existing approaches using extensive simulations.
\end{abstract}

\section{Introduction}
Quality-of-service-based scheduling has received much attention recently. It is shown in \cite{Lai20131689}, \cite{ewaisha2015joint} and \cite{piro2011two} that quality-of-service-aware scheduling results in a better performance compared to best-effort techniques. For example, real-time audio and video applications require algorithms that take hard deadlines into consideration. This is because if a real-time packet is not transmitted on time, the corresponding user might experience intermittent connectivity to its audio or video.

The problem of scheduling for wireless systems under hard-deadline constraints has been widely studied in the literature (see, e.g., \cite{hou2011survey} and \cite{radhakrishnan2016review} for a survey).
%(see e.g., \cite{A_Theory_of_QoS,hou2010scheduling,shakkottai2002scheduling,kang2013performance,Elastic_Inelastic,hou2010utility,Ewaisha_TVT2015}, and the references therein)
 In \cite{A_Theory_of_QoS} the authors consider binary erasure channels and present a sufficient and necessary condition to determine if a given problem is feasible. The work is extended in \cite{hou2010scheduling} to consider general channel fading models. Unlike the time-framed assumption in these works, the authors of \cite{kang2013performance} assume that arrivals and deadlines do not have to occur at the edges of a time frame. In \cite{Elastic_Inelastic} the authors study the scheduling problem in the presence of real-time and non-real-time data. Unlike real-time (RT) data, non-real-time (NRT) data do not have strict deadlines but have an implicit stability constraint on the queues.

Power allocation has not been considered for RT users in the literature, to the best of our knowledge, except in \cite{Ewai1703:Power} that considers on-off fading channels. In this paper, we study a throughput maximization problem in a downlink cellular system serving RT and NRT users simultaneously. We formulate the problem as a joint scheduling-and-power-allocation problem to maximize the sum throughput of the NRT users subject to an average power constraint on the base station (BS), as well as a QoS constraint for each RT user. This QoS constraint requires a minimum ratio of packets to be transmitted by a hard deadline, for each RT user. Perhaps the closest to our work are references \cite{Elastic_Inelastic} and \cite{Ewaisha_TVT2015}. The former does not consider power allocation, while the latter assumes that only one user can be scheduled per time slot. The contributions in this paper are as follows:
\begin{itemize}
	\item We present closed-form expressions for the power allocation policy. It is shown that the power allocation expressions for the RT and NRT users have a different structure.
	\item We present an optimal algorithm satisfying the average power constraint as well as the QoS constraint. We show through simulations that the complexity, in the number of users, of the proposed algorithm is close-to-linear.
\end{itemize}

More details on the results of this work is presented in \cite{Ewaisha_TVT2017}. The rest of this paper is organized as follows. In Section \ref{Model} we present the system model and the underlying assumptions. The problem is formulated in Section \ref{Problem_Formulation} and our optimal algorithm is proposed in Section \ref{Proposed_Algorithm}. Simulation results are presented in Section \ref{Results}. Finally, the paper is concluded in Section \ref{Conclusion}.

%%In this manuscript, we use bold to indicate vectors ${\bf X}$, and calligraphic font to indicate sets $\script{X}$. All logarithms are to the natural base $e$. We use $x^+$ to indicate $\max(x,0)$, $x^*$ to indicate the optimum power of $x$, $\vert \script{X}\vert$ for the cardinality of the set $\script{X}$, $\EE{\cdot}$ to indicate the expected value and $\E_{\bf X}\left[\cdot\right]$ for the expectation conditioned on the random vector ${\bf X}$.

\section{System Model}
\label{Model}
We assume a time slotted downlink system with slot duration $\Ts$ seconds. The system has a single base station (BS) having access to a single frequency channel. There are $N$ users in the system indexed by the set $\sN\triangleq\{1, \cdots,N\}$. The set of users is divided into the RT users $\sNR\triangleq\{1,\cdots,\NR\}$, and NRT users $\sNNR\triangleq\{\NR+1,\cdots, N\}$ with $\NR$ and $\NNR\triangleq N-\NR$ denoting the number of RT and NRT users, respectively. Following \cite{A_Theory_of_QoS}, we model the channel between the BS and the $i$th user as a fading channel with power gain $\gammaik\in [0,\gammamax]$ where $\gammamax<\infty$ is the maximum channel gain that $\gammaik$ can take during the $k$th slot. Channel gains are fixed over the whole slot and change independently in subsequent slots and are independent across users. Moreover, the channel state information for all users are known to the BS at the beginning of each slot in a channel estimation technique that is out of the scope of this paper. The reader is referred to, for example, \cite{bari2015recognizing} on signal classification techniques that precede the channel estimation phase if the modulation scheme was unknown.
 %the BS has to decide whether to admit this packet to the $i$th buffer or drop it out of the system.

\subsection{Packet Arrival Model}
Let $a_i(k)\in\{0,1\}$ be the indicator of a packet arrival for user $i\in\sN$ at the beginning of the $k$th slot. $\{a_i(k)\}$ is assumed to be a Bernoulli process with rate $\lambda_i$ packets per slot and assumed to be independent across all users in the system. Packets arriving at the BS for the RT users are called real-time packets. RT packets have a strict transmission deadline. If an RT packet is not transmitted by this deadline, this packet is dropped out of the system and does not contribute towards the throughput of the user. However, RT user $i$ is satisfied if it receives, on average, more than $q_i\%$ of its total number of packets. We refer to this constraint as the QoS constraint for user $i$. Here we assume that real-time packets arriving at the beginning of the $k$th slot have their deadline at the end of this slot.

On the other hand, packets arriving to the BS for the NRT users can be transmitted at any point in time. Thus, packets for NRT user $i$ are stored, at the BS, at user $i$'s  (infinite-sized \cite{Bertsekas_Data_Networks}) buffer and served on a first-come-first-serve basis. Since the arrival rate $\lambda_i$, for NRT user $i$, might be higher than what the system can support, we define $r_i(k)$ as an admission controller for user $i$ at slot $k$. At the beginning of slot $k$, the BS sets $r_i(k)$ to $1$ if the BS decides to admit user $i$'s arrived packet to the buffer, and to $0$ otherwise. The time-average number of packets admitted to user $i$'s buffer is $\Ai\triangleq \limsup_{K\rightarrow \infty}\frac{1}{K}\sum_{k=1}^K \EE{r_i(k)}$ for all $i\in\sNNR$. 
%\begin{equation}
%\Ai\triangleq \limsup_{K\rightarrow \infty}\frac{1}{K}\sum_{k=1}^K \EE{r_i(k)}, \hspace{0.25in} i\in\sNNR.
%\label{Avg_Admit}
%\end{equation}
And the queue associated with NRT user $i$ is given by
\begin{equation}
Q_i(k+1)=\lb \Qik + Lr_i(k)-\muRik\Rik\rb^+,
\label{Queues}
\end{equation}
$i\in\sNNR$, where $r_i(k)$ is the admission control decision variable for NRT user $i$ at the beginning of slot $k$. We note that no admission controller is defined for the RT users since their buffers cannot build up due to the presence of a deadline.

\subsection{Service Model}
\begin{figure}%
\centering
\includegraphics[width=0.9\columnwidth]{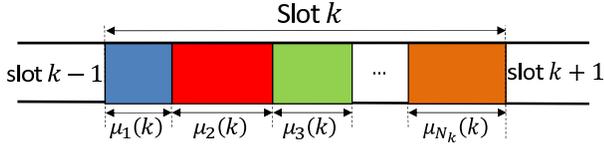}%
\caption{In the $k$th time slot, the BS chooses $\Nk\triangleq\vert\SRk\cup\SNRk\vert$ users to be scheduled. All time slots have a fixed duration of $\Ts$ seconds.}%
\label{Time_Slot}%
\end{figure}
Following \cite{hou2010scheduling} we assume that more than one user can be scheduled in one time slot. However, due to the existence of a single frequency channel in the system, the BS transmits to the scheduled users sequentially as shown in Fig. \ref{Time_Slot}. At the beginning of the $k$th slot, the BS selects a set of RT users denoted by $\SRk\subseteq\sNR$ and a set of NRT users $\SNRk\subseteq\sNNR$ to be scheduled during slot $k$. Moreover, the BS assigns an amount of power $\PRik$ for every user $i\in\sNk$. This dictates the transmission rate for each user according to the channel capacity given by
\begin{equation}
\Rik=\log \lb 1+\PRik\gammaik\rb.
\label{Rate_On_Off}
\end{equation}
Finally, the BS determines the duration of time, out of the $\Ts$ seconds, that will be allocated for each scheduled user. We define the variable $\muRik$ to represent the duration of time, in seconds, assigned for user $i\in\sN$ during the $k$th slot (Fig. \ref{Time_Slot}). Hence, $\muRik\in[0,\Ts]$ for all $i\in\sN$. The BS decides the value of $\muRik$ for each user $i\in\sN$ at the beginning of slot $k$. Since RT users have a strict deadline, then if an RT user is scheduled at slot $k$, then it should be allocated the channel for a duration of time that allows the transmission of the whole packet. Thus we have
%, and Thus, for all $i\in\sNR$, we have $\muRik\in\{0,1\}$ which is a decision variable that represents whether user $i\in\SRk$ or not. For simplicity we define the number of time slots assigned for user $i$ during the $k$th slot to be
\begin{equation}
\muRik=\left\{
\begin{array}{lll}
	\frac{L}{\Rik} &\mbox{ if }i\in\SRk\\
	0 &\mbox{ if }i\in\sNR\backslash\SRk
\end{array}
\right.,
\label{Num_Slots}
\end{equation}where $L$ is the number of bits per packet, that is assumed to be fixed for all packets in the system. Equation \eqref{Num_Slots} means that, depending on the transmission power, if RT user $i$ is scheduled at slot $k$, then it is assigned as much time as required to transmit its $L$ bits. Hence, unlike for the NRT users where $\muRik\in[0,T]$, $\muRik$ is further restricted to the set $\{0,L/\Rik\}$ for the RT users. For ease of presentation, we denote $\bfQ(k)\parNRdef{Q}$. In the next section we present the problem formally.

\section{Problem Formulation}
\label{Problem_Formulation}
We are interested in finding the scheduling and power allocation algorithm that maximizes the sum-rate of all NRT users subject to the system constraints. In this paper we restrict our search to slot-based algorithms which, by definition, take the decisions only at the beginning of the time slots.

Now define the average rate of user $i\in\sNNR$ to be $\bmuRi\triangleq \liminf_{K\rightarrow \infty}\sum_{k=1}^K\muRik \Rik/(L\Ts K)$ packets per slot. Thus the problem is to find the scheduling, power allocation and packet admission decisions at the beginning of each slot, that solve the following problem
\begin{eqnarray}
%\begin{array}{lll}
\label{Prob_DL}
&\text{maximize}&\sum_{i\in\sNNR}\bmuRi,\\
%\end{array}
&\text{subject to } &r_i(k)\leq a_i(k) \forall i\in\sNNR,%\tag{C1}
\label{Admission_Decision}\\
%&\limsup_{K\rightarrow\infty}\frac{1}{K}\EE{Q_i(K)}<\infty, \hspace{0.25in} i\in\sNNR
&& \limsup_{k\rightarrow\infty}\EE{\Qik}<\infty\forall i\in\sNNR,%\tag{C2}
%& \bmuRi\geq\Ai &\forall i\in\sNNR,%\tag{C2}
\label{NRT_QoS}\\
&&\bmuRi\geq\lambda_i q_i \forall i\in\sNR,%\tag{C3}
\label{RT_QoS}\\
&&\limsup_{K\rightarrow \infty}\sum_{k\geq1,i\in\sN}\frac{\PRik \muRik}{K\Ts}\leq \Pavg,%\tag{C4}
\label{P_avg}\\
&& 0\leq\PRik\leq \Pmax \forall i\in\sN,%\tag{C5}
\label{P_max}\\
&&\sum_{i\in\sN}\muRik= \Ts \forall k\geq 1,%\tag{C6}
\label{Single_Tx_at_a_Time}
%&&0\leq\muRik\leq\Ts \forall i\in\sN,%\tag{C7}
%\label{muRik_Constr_NRT}
%\nonumber\text{variables } &,
\end{eqnarray}
where the decision variables are $\bfmuRk\parPRdef{\mu}$, $\bfPRk\parPRdef{P}$ and $\bfrk\parPNRdef{r}$, $\forall k\geq1$. Constraint \eqref{Admission_Decision} says that no packets should be admitted to the $i$th buffer if no packets arrived for user $i$. Constraint \eqref{NRT_QoS} means that the queues of the NRT users have to be stable. Constraint \eqref{RT_QoS} is the RT users' QoS constraint. Constraint \eqref{P_avg} is an average power constraint on the BS transmission power. Finally, constraint \eqref{Single_Tx_at_a_Time} guarantees that the sum of durations of transmission of all scheduled users does not exceed the slot duration $\Ts$. In this paper, we assume that the scheduled NRT user has enough packets, at each slot, to fit the whole slot duration which is a valid assumption in the heavy traffic regime.% It will be clear that the generalization to the non-heavy traffic regime is possible by allowing multiple NRT users to be scheduled but this is omitted for brevity.

\section{Proposed Algorithm}
\label{Proposed_Algorithm}
We use the Lyapunov optimization technique \cite{li2011delay} to find an optimal algorithm that solves \eqref{Prob_DL}. We do this on three steps: i) We define, in Section \eqref{Prob_Decouple} a ``virtual queue'' associated with each average constraint in problem \eqref{Prob_DL}. This helps in decoupling the problem across time slots. ii) In Section \ref{Motivation_DL}, we define a Lyapunov function, its drift and a, per-slot, reward function. iii) Based on the virtual queues and the Lyapunov function, we form and solve an optimization problem, for each slot $k$, that minimizes the drift-minus-reward expression. The solution of this problem is the proposed power allocation and scheduling algorithm.% In Section \ref{Efficient_Solution}, we propose an efficient way to solve this problem optimally. iv) Finally, we show that this minimization guarantees reaching an optimal solution for \eqref{Prob_DL}, in Section \ref{Proposed_Algorithm_Cont}.

%We solve this problem using the Lyapunov optimization technique \cite{li2011delay}, where each average constraint in problem \eqref{Prob_DL} is associated with a ``virtual queue'' which reflects how much this average constraint is violated.
\subsection{Problem Decoupling Across Time Slots}
\label{Prob_Decouple}
We define a virtual queue associated with each RT user as follows
\begin{equation}
Y_i(k+1)=\lb \Yik + a_i(k)q_i-\Oneik\rb^+, \hspace{0.25in} i\in\sNR,
\label{DL_VQ}
\end{equation}
where $\Oneik\triangleq\mathds{1}\lb\muRik\rb$ with $\mathds{1}(\cdot)=1$ if its argument is non-zero and $\mathds{1}(\cdot)=0$ otherwise. For notational convenience we denote ${\bfY}(k)\parRdef{Y}$. $\Yik$ is a measure of how much constraint \eqref{RT_QoS} is violated for user $i$. We will later show a sufficient condition on $Y_i(k)$ for constraint \eqref{RT_QoS} to be satisfied. Hence, we say that the virtual queue $Y_i(k)$ is associated with constraint \eqref{RT_QoS}. Similarly, we define the virtual queue $X(k)$, associated with constraint \eqref{P_avg}, as
\begin{equation}
X(k+1)=\lb X(k) + \frac{\sum_{i\in\sN}\PRik \muRik}{\Ts}-\Pavg\rb^+.
\label{P_avg_VQ}
\end{equation}
To provide a sufficient condition on the virtual queues to satisfy the corresponding constraints, we use the definition of \emph{mean rate stability} of queues \cite[Definition 1]{li2011delay} to state the following lemma.

%\begin{Def}
%\label{Mean_Rate_Def}
%A random sequence $\{Y_i(k)\}_{k=0}^\infty$ is said to be mean rate stable if and only if $\limsup_{K\rightarrow\infty}\EE{Y_i(K)}/K=0$ holds.
%\end{Def}

\begin{lma}
\label{Mean_Rate_Lemma}
If, for some $i\in\sNNR$, $\{Y_i(k)\}_{k=0}^\infty$ is mean rate stable, then constraint \eqref{RT_QoS} is satisfied for user $i$.
\end{lma}
%\begin{proof}
%Proof follows along the lines of Lemma 3 in \cite{li2011delay}.
%\end{proof}
Lemma \ref{Mean_Rate_Lemma} shows that when the virtual queue $\Yik$ is mean rate stable, then constraint \eqref{RT_QoS} is satisfied for user $i\in\sNNR$. Similarly, if $\{X(k)\}_{k=0}^\infty$ is mean rate stable, then constraint \eqref{P_avg} is satisfied. Thus, our objective would be to devise an algorithm that guarantees the mean rate stability of both $\left[Y_i(k)\right]_{i\in\sNR}$ and $X(k)$.

\subsection{Applying the Lyapunov Optimization}
\label{Motivation_DL}
The quadratic Lyapunov function is defined as
\begin{equation}
L_{\rm yap}\lb U(k)\rb\triangleq \frac{1}{2}\sum_{i\in\sNR}{Y_i^2(k)}+\frac{1}{2}\sum_{i\in\sNNR}{Q_i^2(k)}+\frac{1}{2}X^2(k),
\label{Lyapunov_Func}
\end{equation}
where $\bfU(k)\triangleq \lb\bfYk,\bfQk,X(k)\rb$, and the Lyapunov drift as $\Delta (k) \triangleq \E_{U(k)}[L_{k+1}\lb {\bf U}(k+1)\rb - L_{\rm yap}\lb \bfU(k)\rb]$ where $\EEU{x}\triangleq \EE{x\vert U(k)}$ is the conditional expectation of the random variable $x$ given $U(k)$. Squaring \eqref{Queues}, \eqref{DL_VQ} and \eqref{P_avg_VQ} taking the conditional expectation then summing over $i$, the drift becomes bounded by
\begin{equation}
\Delta(k)\leq \frac{C_1}{2}+\Psi(k),
\label{Drift_Bound}
\end{equation}
where $C_1\triangleq\sum_{i\in\sNR}\lb q_i^2+1\rb+\Pmax^2+\Pavg^2+\NNR L^2+\NNR\Ts^2\Rmax^2$ and we use $\Rmax\triangleq\log\lb1+\Pmax\rb$, while
\begin{align}
%\begin{multline}
\nonumber\Psi(k)\triangleq &\sum_{i\in\sNR}\EEU{\Yik\lb \lambda_i q_i-{\Oneik}\rb}\\
\nonumber+&X(k)\lb\sum_{i\in\sN}\frac{\EEU{\muRik\PRik}}{\Ts}-\Pavg\rb\\
+&\sum_{i\in\sNNR}\Qik\lb \EEU{Lr_i(k)-\muRik\Rik}\rb.
\label{Psi_k}
\end{align}
We define $\Bmax$ as an arbitrarily chosen positive control parameter that controls the performance of the algorithm. We shall discuss the tradeoff on choosing $\Bmax$ later on. Since $\EEU{Lr_i(k)}$ represents the average number of bits admitted to NRT user $i$'s buffer at slot $k$, we refer to $\Bmax \sum_{i\in\sNNR}\EEU{Lr_i(k)}$ as the ``reward term''. We subtract this term from both sides of \eqref{Drift_Bound}, then use \eqref{Psi_k} and rearrange to bound the drift-minus-reward term as
\begin{multline}
\Delta(k)-\Bmax \sum_{i\in\sNNR}\EEU{Lr_i(k)}\leq C_1-X(k)\Pavg\\+\EEU{\sum_{i\in\sNR}\PsiRik}+\EEU{\sum_{i\in\sNNR}\PsiNRik\muRik}\\
+\EEU{\sum_{i\in\sNNR}\lb\Qik-\Bmax \rb Lr_i(k)}+\sum_{i\in\sNR}\Yik\lambda_i q_i,
\label{Drift_minus_Reward_Bound}
\end{multline}
%\begin{equation}
%\begin{array}{ll}
	%&\text{maximize}\sum_{i\in\SRk}\PsiRik + \sum_{i\in\sNNR}\PsiNRik\muRik+\sum_{i\in\sNNR}\lb\Qik-\Bmax \rb Lr_i(k)\\
	%&\text{subject to } \eqref{P_max}, \eqref{Single_Tx_at_a_Time} \text{ and } \eqref{muRik_Constr_NRT}.
%\end{array}
%\label{Max_Prob_b4_Decouple}
%\end{equation}
%where in the last equation the scheduler $\muRik$ is set to $0$ for all $i\in\sN$ having $\gammaik=0$ and for all $i\in\sNNR$ having $\Qik=0$ since these users cannot be scheduled.
where $\PsiRik\triangleq \lb\Yik-\frac{L}{\Ts\Rik}X(k)\PRik\rb\Oneik$ for all $i\in\sNR$ and $\PsiNRik\triangleq \Qik\Rik-\frac{X(k)\PRik}{\Ts}$ for all $i\in\sNNR$
%\begin{align}
%&\PsiRik\triangleq \lb\Yik-\frac{L}{\Ts\Rik}X(k)\PRik\rb\Oneik, &i\in\sNR,
%\label{PsiRik}\\
%&\PsiNRik\triangleq \Qik\Rik-\frac{X(k)\PRik}{\Ts}, &i\in\sNNR,
%\label{PsiNRik}
%\end{align}
The proposed algorithm schedules the users, allocates their powers and controls the packet admission to minimize the right-hand-side of \eqref{Drift_minus_Reward_Bound} at each slot. Since the only term in right-hand-side of \eqref{Drift_minus_Reward_Bound} that is a function in $r_i(k)$ $\forall i\in\sNNR$ is the fourth term, we can decouple the admission control problem from the joint scheduling-and-power-allocation problem. Minimizing this term results in the following admission controller: set $r_i(k)=a_i(k)$ if $Q_i(k)<\Bmax$ and $0$ otherwise. Minimizing the remaining terms yields
\begin{equation}
\begin{array}{ll}
	&\text{maximize}\sum_{i\in\SRk}\PsiRik + \sum_{i\in\sNNR}\PsiNRik\muRik\\
	&\text{subject to } \eqref{P_max} \text{ and } \eqref{Single_Tx_at_a_Time}.% \text{ and } \eqref{muRik_Constr_NRT}.
\end{array}
\label{Max_Prob}
\end{equation}
with decision variables $\bfPRk$ and $\bfmuRk$. This is a per-slot optimization problem the solution of which is an algorithm that minimizes the upper bound on the drift-minus-reward term defined in \eqref{Drift_minus_Reward_Bound}. Next we show how to solve this problem in an efficient way.
\vspace{-0.1in}
\subsection{Efficient Solution for the Per-Slot Problem}
\label{Efficient_Solution}
To solve this problem optimally, we first find the optimal power-allocation-and-scheduling policy for the NRT users through the following lemma.

\begin{lma}
\label{NRT_Lemma_Cont}
If user $i\in\sNNR$ is scheduled to transmit any of its NRT data during the $k$th slot, then the optimum power level for this NRT w.r.t. problem \eqref{Max_Prob} in the continuous fading case is given by
\begin{equation}
\PNRik=\min\lb\lb\frac{Q_i(k)}{X(k)}-\frac{1}{\gammaik}\rb^+,\Pmax\rb.
\label{H2O_Pow_Cont}
\end{equation}
Moreover, in the heavy traffic regime, the scheduled NRT user, if any, that optimally solves problem \eqref{Prob_DL} is $\ist=\arg\max_{i\in\sNNR}\PsiNRikst$ with ties broken randomly uniformly, while
$\PsiNRikst\triangleq \Qik\log \lb\Qik\rb-\Qik+\frac{X(k)}{\gammaik}-\Qik\log\lb \frac{X(k)}{\gammaik}\rb$.
\end{lma}
\begin{proof}
The proof is omitted for brevity.
\end{proof}

Lemma \ref{NRT_Lemma_Cont} presents the optimal power and scheduling policy for the NRT users. To solve for the scheduling and power allocation for the RT users, we first solve for $\PRik$ assuming a fixed subset $\SRk\subseteq\sNR$, then, the optimum set $\SRkst$ is the one that maximizes \eqref{Max_Prob}. The expression for $\PRik$ for the RT users is one of the main contributions of this paper and is presented in the following theorem.

\begin{thm}
\label{PRik_Lemma_Cont}
In the continuous-fading channel model, given some non-empty set $\SRk$, the power allocation policy
\begin{equation}
\PRik=\min\lb\frac{1}{\gammaik}\left[\frac{\tilde{\phi}\gammaik-1}{\LW{\left[\tilde{\phi}\gammaik-1\right]}}-1\right],\Pmax\rb,
\label{Lambert_Pow}
\end{equation}
$i\in\SRk$ with $\tilde{\phi}\triangleq\lb\PsiNRistkst+\phi\rb\Ts/X(k)$, is optimal w.r.t. \eqref{Max_Prob} when $\phi$ is set to a non-negative value that satisfies \eqref{Single_Tx_at_a_Time}.
\end{thm}
\begin{proof}
See \cite{Ewaisha_TVT2017} for the complete proof.
\end{proof}

It is clear that the Lambert power policy in \eqref{Lambert_Pow} has a different structure than the water-filling policy in \eqref{H2O_Pow_Cont}. The reason is because the former is for transmitting packet that have hard deadlines. The following theorem, stated without proof due to lack of space, discusses the monotonicity of the Lambert power policy.
\begin{thm}
\label{Thm_Lambert_Monotonicity}
Let $\SRk$ be some scheduling RT set at slot $k$. The power $\PRik$ given by \eqref{Lambert_Pow} is monotonically decreasing in $\gammaik$ $\forall i\in \SRk$.
\end{thm}

In \cite{Ewaisha_TVT2017}, we plot \eqref{Lambert_Pow} and \eqref{H2O_Pow_Cont} versus $\gammaik$ to contrast the fact that, while the water-filling is an increasing function in the channel gain, the Lambert is a decreasing function in the channel gain. This is because the RT user has a single packet of a fixed length to be transmitted. If the channel gain increases, then the power decreases to keep the same transmission rate resulting in the same transmission duration of one slot.
%\begin{figure}
	%\centering
		%\includegraphics[width=0.9\columnwidth]{figures/Plot_H2O_vs_Lambert.eps}
		%\caption{The Lambert power policy decreases with the channel gain, while the water-filling policy increases with the gain.}
	%\label{Plot_H2O_vs_Lambert}
%\end{figure}

%\begin{proofsketch}
%Proof follows by differentiating \eqref{Lambert_Pow} with respect to $\gammaik$ for some user $i$, while having $\phi$ satisfying \eqref{Single_Tx_at_a_Time}, and showing that the resulting derivative is always non-positive for $\gammaik\geq 0$.
%\end{proofsketch}

The optimum scheduling algorithm for the RT users is to find, among all subsets of the set $\sNR$, the set that gives the highest objective function of \eqref{Max_Prob}.% Thus the complexity is of $O\lb2^{\NR}\rb$. In Section \ref{Heuristic} we present a heuristic with a lower complexity.

%\subsection{Proposed Algorithm and Proof of Optimality}
%\label{Proposed_Algorithm_Cont}
The following theorem is stated as an effort to achieve an algorithm with a relatively small complexity.%Observing the approach in the special on-off case and inspired by Theorem \ref{NRik_Lemma_On_Off} that reduces the search space, we provide here a similar approach. We first provide the following definition which is analogous to Theorem \ref{NRik_Lemma_On_Off}.
%\begin{Def}
%\label{Eliminated_Sets}
%At slot $k$, the set $\SRk$ is said to be an ``eliminated'' set if there exists some $i\notin\SRk$ and some $j\in\SRk$ such that $\Yik>Y_j(k)$ and $\gammaik>\gamma_j(k)$.
%\end{Def}

\begin{thm}
\label{NRik_Lemma_Cont}
At slot $k$, for any set $\SRk$, if there exists some $i\notin\SRk$ and some $j\in\SRk$ such that $\Yik>Y_j(k)$ and $\gammaik>\gamma_j(k)$, then $\SRk$ cannot be an optimal RT set, with respect to problem \eqref{Max_Prob}, for the continuous channel model.
\end{thm}
\begin{proof}
See \cite{Ewaisha_TVT2017} for the complete proof.
\end{proof}
This theorem provides a sufficient condition for non-optimality. In other words, we can make use of this theorem to restrict our search algorithm to the sets that do not satisfy this property. Before presenting the proposed algorithm, we define the set $\SRT$ as the set of all possible subsets of the set $\sNR$.%, and define, for all $i\in\sNNR$, the function
%\begin{equation}
%\PsiNRik\triangleq \Qik\log\lb 1+\PRik\rb- X(k)\PRik.
%\label{PsiNRik_Cont}
%\end{equation}
\begin{algorithm}
\caption{Lambert-Strict Algorithm}
\begin{algorithmic}[1]
\label{Scheduling_Alg_Cont}
\STATE Define the auxiliary functions $\Psiaux(\cdot):\SRT\rightarrow \mathbb{R}_+$ and $\Paux(\cdot,\cdot):\SRT\times\sNR\rightarrow\mathbb{R}_+$.
\STATE Initialize $\Paux(\sS,i)=0$ for all $\sS\in\SRT$ and all $i\in\sNR$.
\STATE Find the user $\ist$ and its power as given in Lemma \ref{NRT_Lemma_Cont}.
\FOR {$\sS\in \SRT$}
\IF{$\exists$ some $i\notin\sS$ and some $j\in\sS$ such that $\Yik>Y_j(k)$ and $\gammaik>\gamma_j(k)$}
\STATE Set $\Psiaux(\sS)=-\infty$ and go to Step 4 (next set in $\SRT$).
\ENDIF
\STATE $\phi\leftarrow\phimax+\Delta\phi$
\WHILE{$\phi\muRik\neq 0$}
\STATE $\phi\leftarrow\phi-\Delta\phi$. Calculate $\PRik$ given by \eqref{Lambert_Pow} for all $i\in\sS$ and set $\muRistk= \Ts-\sum_{i\in\sS}\muRik$.
%\begin{equation}
%\muRistk= \Ts-\sum_{i\in\sS}\muRik
%\label{Single_Tx_at_a_Slot_Cont}
%\end{equation}
\ENDWHILE
\STATE Set $\Psiaux(\sS)=\sum_{i\in\sS} \lb Y_i(k)-X_i(k)\muRik\rb + \PsiNRistkst\muRistk$ and $\Paux(\sS,i)=\PRik$ $\forall i\in\sS$.
\ENDFOR
\STATE The scheduling set is $\SRstk=\arg\max_\sS\Psiaux(\sS)$.
\STATE Set $\PRikst=\Paux\lb\SRstk,i\rb$ for all $i\in\sNR$, set $\mu_{\ist}(k)=\Ts-\sum_{i\in\SRstk}\muRik$ and set $r_i(k)=a_i(k)$ if $Q_i(k)<\Bmax$ and $0$ otherwise $\forall i\in\sNNR$.
\STATE Update \eqref{Queues}, \eqref{DL_VQ} and \eqref{P_avg_VQ} at the end of the $k$th slot.
\end{algorithmic}
\end{algorithm}

\begin{thm}
\label{Optimality_Thm_Cont}
For the continuous channel model, if problem \ref{Prob_DL} is feasible, then for any $\Bmax >0$ Algorithm \ref{Scheduling_Alg_Cont} satisfies all constraints in \eqref{Prob_DL} and achieves an average sum throughput satisfying
\begin{equation}
\sum_{i\in\sNNR} \bmuRi\geq \sum_{i\in\sNNR}{\bmuRi^*} - \frac{C_1}{L\Bmax },
\label{Optimality_Eq_Cont}
\end{equation}
where $\bmuRi^*$ is the optimal rate for user $i$ w.r.t. \eqref{Prob_DL}.
\end{thm}

\begin{proof}
See \cite{Ewaisha_TVT2017} for the complete proof.
\end{proof}

% The algorithm would be a slight modification on Algorithm \ref{Scheduling_Alg_Cont}. That is, we add an if-conditional at the beginning of the for-loop, namely, between lines 4 and 5. The statement is: {\textbf{IF}} there exists some $i\notin\sS$ and some $j\in\sS$ such that $\Yik>Y_j(k)$ and $\gammaik>\gamma_j(k)$, {\textbf{THEN}} skip this set and continue to the next set in $\SRT$, \textbf{END}. We refer to this algorithm as the ``Lambert-Strict'' Algorithm.

Due to the problem being a combinatorial problem with a huge amount of possibilities, we could not reach a closed-form expression for the complexity order of this algorithm. However, simulations will show its complexity improvement over the exhaustive search algorithm.

\section{Simulation Results}
\label{Results}
We simulate the system assuming that all channels are statistically homogeneous, i.e. $\bgamma_i=1$ for all $i\in\sN$. Moreover, all RT users have homogeneous QoS constraints, thus $q_i=q$ for all $i\in\sNR$ for some parameter $q$. All parameter values used in the simulations are: $L=1$ bit, $\Bmax=100$, $\Ts=5$, $\Pavg=10$, $q=0.9$ and $\Pmax=20$. In Fig. \ref{Cont_Opt_vs_LambStrict_Complexity}, we plot the complexity of the Lambert-Strict algorithm as well as the exhaustive search algorithm with exponential complexity versus the number of users $\NR$. The complexity is measured in terms of the average number of iterations, per-slot, where we have to evaluate the objective function of \eqref{Max_Prob}. Since this complexity changes from a slot to the other, we plot the average of this complexity. As the number of users increases, the Lambert-Strict algorithm has an average complexity close to linear.% We note that we simulated this system with . As the number of users increase more RT users are scheduled, on average, per time slot. This comes at the expense of the NRT users' throughput. Although the number of NRT users increase as well creating multi-user diversity effect, we do not observe an increase in the throughput. This is because NRT users are not scheduled based on the channel gain only but on the queue length as well.
\begin{figure}%
\centering
\includegraphics[width=0.78\columnwidth]{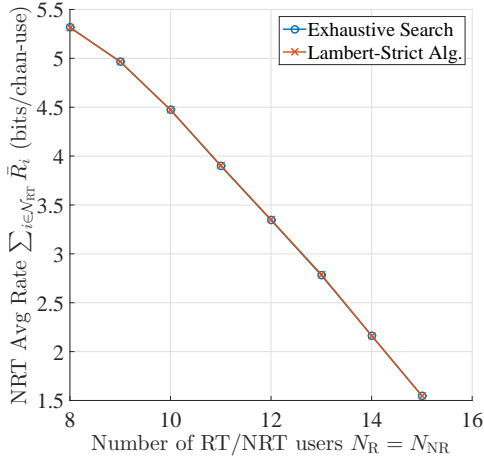}%
\caption{The Lambert-Strict Algorithm yields the same throughput as the exhaustive search algorithm but with a lower average complexity.}%
\label{Cont_Opt_vs_LambStrict_Throughput}%
\end{figure}

\begin{figure}%
\centering
\includegraphics[width=0.78\columnwidth]{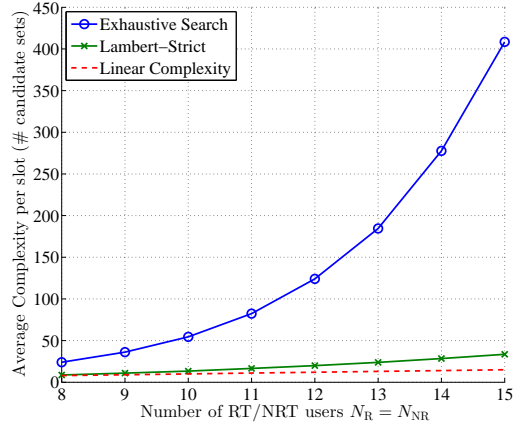}%
\caption{As the number of NRT users in the system increase the complexity increases exponentially for exhaustive search and nearly linear for the Lambert-Strict algorithm.}%
\label{Cont_Opt_vs_LambStrict_Complexity}%
\end{figure}

%\vspace{-0.3in}
\section{Conclusions}
\label{Conclusion}
We discussed the problem of throughput maximization in downlink cellular systems in the presence of RT and NRT users. We formulated the problem as a joint power-allocation-and-scheduling problem. Using the Lyapunov optimization theory, we presented an optimal algorithm that solves the constrained throughput maximization problem. The complexity of the proposed algorithm is shown, through simulations, to have a close-to-linear complexity. Moreover, the power allocations are presented in closed-form expressions for the RT as well as the NRT users. We showed that the NRT power allocation is water-filling-like which is monotonically increasing in the channel gain. On the other hand, the RT power allocation has a totally different structure that we call the ``Lambert Power Allocation''. It is found that the latter is a decreasing function in the channel gain.

%\appendices
%
%\section{Proof of Theorem \ref{Optimality_Thm}}
%\input{App_Optimality_Thm_Proof}

%----------------------------------------------------------------------------------

%--------------------------------      References      ----------------------------

%----------------------------------------------------------------------------------
\bibliographystyle{IEEEbib}
\bibliography{MyLib}

\end{document}